\documentclass[11pt]{article}

\usepackage{mathpazo}
\usepackage{amsfonts}
\usepackage{amsmath}
\usepackage{amssymb}
\usepackage{amsthm}
\usepackage{graphicx}
\usepackage{latexsym}
\usepackage[margin=1.05in]{geometry}
\usepackage[T1]{fontenc}
\usepackage{color,graphicx}
\usepackage{enumerate}
\usepackage{xcolor}
\usepackage[colorlinks=true,linkcolor=blue,anchorcolor=blue,citecolor=red,urlcolor=magenta]{hyperref}
\usepackage{caption}

\newtheorem{theorem}{Theorem}[section]
\newtheorem{corollary}{Corollary}[section]

\newtheorem{proposition}{Proposition}[section]

\newcommand{\tr}{\operatorname{Tr}}

\newcommand{\defeq}{\stackrel{\smash{\textnormal{\tiny def}}}{=}}

\def\C{\mathbb{C}}
\def\R{\mathbb{R}}

\def\D{{\mathcal D}}
\def\M{{\mathcal M}}
\def\diag{{\rm diag}}

\begin{document}
\title{Approximating Quantum States by States of Low Rank}

\author{
	Nathaniel Johnston\textsuperscript{1} and Chi-Kwong Li\textsuperscript{2}
}

\date{January 23, 2026}

\maketitle

\begin{abstract}
    Given a positive integer $k$, it is natural to ask for a formula for the distance between a given density matrix (i.e., mixed quantum state) and the set of density matrices of rank at most $k$. This problem has already been solved when ``distance'' is measured in the trace or Frobenius norm. We solve it for all other unitary similarity invariant norms. We also present some consequences of our formula. For example, in the trace and Frobenius norms, the density matrix that is farthest from the set of low-rank density matrices is the maximally-mixed state, but this is not true in many other unitary similarity invariant norms.\\

	
\end{abstract}

\addtocounter{footnote}{1}
\footnotetext{Department of Mathematics \& Computer Science, Mount Allison University, Sackville, NB, Canada E4L 1E4}

\addtocounter{footnote}{1}
\footnotetext{Department of Mathematics, The College of William and Mary, Williamsburg, VA, USA 23185}

\section{Introduction}\label{sec:intro}

In quantum information theory, the rank of a quantum state (i.e., density matrix) is a fundamental quantity that provides a discrete measure of ``how mixed'' that state is. For most applications, pure (i.e., rank~$1$) states are the most desirable and are the easiest to manipulate (see \cite{Mer98}, \cite{SKW18}, \cite[Section~1.1]{Jae07}, and the references therein), though low-rank states in general also have computational advantages \cite{EBA23,GS24,SSS25}. For this reason it seems natural to ask, given a quantum state, what is the closest low-rank state to it?

There are many different norms that we could use to measure the distance between quantum states, and a priori it seems possible that different norms may lead to different low-rank approximations. The most common method of measuring distance between quantum states is to use the trace norm, but the Frobenius, operator, Schatten \cite[Section~1.1.3]{Wat18}, and Ky Fan norms \cite{CK09} all make frequent appearances in the area as well, as do the numerical radius and several others \cite{JK15}. All of these norms have a property called ``unitary similarity invariance'': they remain unchanged under conjugation by an arbitrary unitary matrix.

If we did not care about the trace of the low-rank approximation of our density matrix, then the well-known Eckart--Mirsky--Young theorem \cite{EY36,Mir60} would tell us that, regardless of which unitary similarity invariant norm is being used, the closest rank (at most)-$k$ approximation is obtained by truncating the matrix's spectral decomposition after the $k$ largest eigenvalues. However, since we require that the rank-$k$ approximation itself be a valid density matrix (and thus have trace $1$), the problem is not so simple.

It was shown in \cite{EHC22} that, for the trace norm and the Frobenius norm, the closest rank-$k$ state is obtained by applying a certain additive scaling to the spectral decomposition truncation. Our main contribution is to show that the same is true in every other unitary similarity invariant norm (Theorem~\ref{thm:low_rank}).

There are numerous potential applications of this main result. For example, low-rank quantum states can be stored and transmitted more efficiently than high-rank quantum states, since a rank-$k$ state can be specified by just $O(kn)$ real parameters (via its spectral decomposition) rather than the $n^2 - 1$ real parameters that are needed to specify an arbitrary state. It thus might be desirable in some limited-resource situations to approximate a quantum state by the nearest low-rank state and work with or transmit that state instead. As another example, in quantum information theory the problem of reconstructing a quantum state based on measurement results is called state tomography. Performing state tomography to reconstruct pure (i.e., rank-$1$) states is a task that has been thoroughly studied \cite{CJZ13,VD23}, as has this task for other low ranks \cite{Ach16}. Because measurements are noisy in practice, tomography often results in a quantum state that is not actually pure even though it should be, and the same is true when tomography is applied to states of other low ranks. Our results provide a way of converting the (incorrect) tomographically-reconstructed state into the nearest state that has the desired rank.

We then explore several consequences of our result, such as the question of which quantum state is farthest from the set of low-rank quantum states, or equivalently what the maximum possible distance between a quantum state and the set of low-rank quantum states is. Distance measures like this are common in the theory of quantum resources~\cite{Reg18}, and it is important to know what values they can take on.

For example, the well-studied geometric measure of entanglement asks for the distance between a bipartite pure state and the set of separable pure states, and this quantity is maximized exactly by the maximally-entangled states \cite{WG03} (i.e., the states with all of their Schmidt coefficients equal to each other). We show that, similarly, the answer of the analogous question of which state is farthest from low-rank is the maximally-mixed state in several cases, such as when the norm being used is the trace norm (Proposition~\ref{prop:max_mixed_trace_norm}). However, we also show that there are many other unitary similarity invariant norms for which this is not true. In general, all that can be said is that the farthest state is one of $n-k$ different candidates (Theorem~\ref{thm:max_mixed_maximizes}), one of which is the maximally-mixed state.

We then further investigate these questions for the Schatten and Ky Fan norms, since they are the most commonly-occurring families of unitary similarity-invariant norms in quantum information theory (see \cite{Wat05,Wat18,HO22} and the references therein, for example). The Schatten norms occur frequently in this area because the trace, Frobenius, and operator norms occur as special cases of them, and because of their connection to the (now-disproved) additivity conjecture \cite{HW08,Has09}. The Ky Fan norms similarly make frequent appearances in the field because of their connection to majorization, which is absolutely ubiquitous in the field \cite{Nie02,Jab13}.

In particular, we show that if $p \in \{1\} \cup [2,4]$, then the farthest state from the set of low-rank states when using the Schatten $p$-norm is always the maximally-mixed state, but for all other values of $p$ there exists a rank $k$ and a dimension $n$ for which the farthest state (in Schatten $p$-norm) is not the maximally-mixed state (Theorem~\ref{thm:schatten_124}). We also determine the farthest state in the Ky Fan norms in Theorem~\ref{thm:ky_fan_extreme}.

\subsection{Preliminaries and notation}\label{sec:prelim}

Let $\M_{n}$ denote the set of $n \times n$ complex matrices. Let $\M_n^{+} \subset \M_n$ denote the set of (Hermitian) positive semidefinite matrices. If $X \in \M_n^{+}$ has $\tr(X) = 1$ then it is called a \emph{quantum state} or a \emph{density matrix}; we denote the set of density matrices in $\M_n^{+}$ by $\D_n$. Given an integer $k \in \{1,2,\ldots,n\}$, let $\D_{n,k} \subseteq \D_n$ denote the set of density matrices of rank at most $k$.

Denote the zero matrix and the identity matrix by $O$ and $I$, respectively, or $O_n$ and $I_n$ if we wish to emphasize that they are $n \times n$. The density matrix $\tfrac{1}{n}I_n$ is called the \emph{maximally-mixed state}. For a vector $(x_1,x_2,\ldots,x_n) \in \C^n$, let $\diag(x_1,x_2,\ldots,x_n) \in \M_n$ denote the diagonal matrix with diagonal entries $x_1,x_2,\ldots,x_n$ (in that order). For matrices $X \in \M_m$ and $Y \in \M_n$, $X \oplus Y \in \M_{m+n}$ is the direct sum of $X$ and $Y$, which we will typically think of as a block-diagonal matrix with diagonal blocks $X$ and $Y$, in that order.

A norm on $\M_n$ is \emph{unitary similarity invariant (USI)} if $\|U^*XU\| = \|X\|$ for all $X \in \M_n$ and all unitary $U \in \M_n$. If $\sigma_1 \geq \sigma_2 \geq \cdots \geq \sigma_n \geq 0$ denote the singular values of $X$ then some standard examples of USI norms include:
\begin{itemize}
    \item The operator norm: $\|X\|_{\textup{op}} = \sigma_1$.

    \item The trace norm: $\|X\|_{\textup{Tr}} = \sum_{j=1}^n \sigma_j$.

    \item The Frobenius norm: $\|X\|_{\textup{F}} = \sqrt{\sum_{j=1}^n \sigma_j^2}$.

    \item The Schatten norms: For real $p \in [1,\infty)$, $\|X\|_p = \left(\sum_{j=1}^n \sigma_j^p\right)^{1/p}$. The trace and Frobenius norms arise in the $p = 1$ and $p = 2$ special cases, respectively. The operator norm also arises in the large-$p$ limit: $\lim_{p\rightarrow\infty}\|X\|_p = \|X\|_{\textup{op}}$. For this reason, we often say $p \in [1,\infty]$ with the understanding that $\|X\|_\infty = \|X\|_{\textup{op}}$.

    \item The Ky Fan norms: For an integer $r \in \{1,2,\ldots,n\}$, $\|X\|_{(r)} = \sum_{j=1}^r \sigma_j$. The operator and trace norms arise in the $r = 1$ and $r = n$ special cases, respectively.
\end{itemize}

Given vectors $\mathbf{x}, \mathbf{y} \in \R^n$, we use $x_j$ and $y_j$ to denote the $j$-th entries, respectively, and we use $x_j^\downarrow$ and $y_j^\downarrow$ to denote their $j$-th largest entries, respectively (e.g., $x_1^\downarrow$ is the largest entry of $\mathbf{x}$, $x_2^\downarrow$ is its next largest entry, and so on). We say that $\mathbf{y}$ \emph{majorizes} $\mathbf{x}$, and we write $\mathbf{x} \prec \mathbf{y}$, if $\sum_{j=1}^r x_j^\downarrow \leq \sum_{j=1}^r y_j^\downarrow$ for all $r \in \{1,2,\ldots,n\}$ and $\sum_{j=1}^n x_j = \sum_{j=1}^n y_j$.

We note that majorization has numerous applications in quantum information theory (see \cite{Nie02} and \cite[Section~4.3]{Wat18}, for example). Our particular interest in it comes from the following result that relates majorization to USI norms:

\begin{theorem}[{\cite[Theorem~2.2]{LT1}}]\label{thm:usi_from_major}
    For a Hermitian matrix $A \in \M_n$, let $\lambda(A) \in \R^n$ denote its sorted vector of eigenvalues (i.e., $\lambda(A) = (a_1, a_2, \ldots, a_n)$, where $a_1 \geq a_2 \geq \cdots \geq a_n$ are the eigenvalues of $A$). Suppose $X, Y \in \M_n$ are Hermitian. The following are equivalent:
    \begin{itemize}
        \item $\|X\| \leq \|Y\|$ for all USI norms.

        \item There exists $t \in [0,1]$ for which $\lambda(X) \prec t\lambda(Y) + (1-t)\lambda(-Y)$.
    \end{itemize}
\end{theorem}

\section{Low-rank approximation of density matrices}\label{sec:main_low_rank}

We now introduce the main problem that we are concerned with. Given integers $1 \leq k \leq n$, a unitary similarity invariant norm $\| \cdot \|$ on $\M_n$, and a density matrix $X \in \D_n$, our goal now is to determine the value of
\[
    d(X,\D_{n,k}) \defeq \min_{Y \in \D_{n,k}}\big\{ \|X - Y \| \big\},
\]
and to find a density matrix $Y \in \D_{n,k}$ attaining this minimum value.

In the case when the norm is the trace or Frobenius norm, this problem was solved in \cite{EHC22}, where it was shown that a closest rank (at most)-$k$ density matrix has the form $Y = P_k X P_k + \gamma P_k$, where $P_k$ is the orthogonal projection onto the direct sum of the eigenspaces corresponding to the $k$ largest eigenvalues of $X$, and $\gamma$ is the unique real scalar that results in $Y$ having trace $1$. We now show that the same is true for all unitary similarity invariant norms.

\begin{theorem}\label{thm:low_rank}
    Suppose $1 \leq k \leq n$ are integers, $\| \cdot \|$ is a unitary similarity invariant norm, and $X \in \D_n$ has spectral decomposition $X = \sum_{j=1}^n x_j \mathbf{v}_j \mathbf{v}_j^*$ with $x_1 \ge \cdots \ge x_n \ge 0$. Define
    \[
        \gamma = \frac{1}{k}\sum_{j=k+1}^n x_j \quad \text{and} \quad Y = \sum_{j=1}^k (x_j + \gamma) \mathbf{v}_j \mathbf{v}_j^*.
    \]
    Then $d(X,\D_{n,k}) = \|X - Y\|$. In other words, $\|X - Y\|\le \|X - Z\|$ for all $Z \in \D_{n,k}$.
\end{theorem}

The orthogonal projection $P_k$ discussed earlier is, in terms of quantities described by this theorem, $P_k = \sum_{j=1}^k \mathbf{v}_j \mathbf{v}_j^*$, resulting in $Y = P_k X P_k + \gamma P_k$ as claimed. Furthermore, this minimizer $Y$ is unique (up to unitary similarity) whenever the USI norm is strictly convex (compare with the result of \cite{EHC22}, which noted that this minimizer $Y$ is unique for the Frobenius norm but not for the trace norm).

\begin{proof}[Proof of Theorem~\ref{thm:low_rank}.]
    Suppose $X$ and $Y$ have the spectral decompositions described by the statement of the theorem. Let $Z \in \D_n$ have eigenvalues $z_1 \ge \cdots \ge z_n$ such that $z_j = 0$ for $j > k$ (i.e., $Z \in \D_{n,k}$). Then 
    \begin{equation}\label{ineq:norm_to_diag}
        \|X - Z\| \ge \|\diag(x_1-z_1, \ldots, x_n-z_n)\| = \|\diag(x_1-z_1, \ldots, x_k-z_k, x_{k+1}, \ldots, x_n)\|,
    \end{equation}
    with the inequality following, for example, from \cite[Equation~(IV.62)]{Bha97}.\footnote{Strictly speaking, that result is stated just for unitarily-invariant norms (i.e., norms for which $\|UXV\| = \|X\|$ whenever $U$ and $V$ are unitary), but the exact same proof works for unitary similarity invariant norms.} Now observe that
    \[
        -\gamma = -\frac{1}{k}\sum_{j=k+1}^n x_j = \left(\frac{1}{k}\sum_{j=1}^k x_j\right) - \frac{1}{k} = \frac{1}{k}\sum_{j=1}^k (x_j - z_j).
    \]
    It follows that $(-\gamma, \dots, -\gamma) \prec (x_1-z_1, \ldots, x_k-z_k)$, so
    \[
        (-\gamma, \ldots, -\gamma, x_{k+1}, \ldots, x_n) \prec (x_1-z_1, \ldots, x_k-z_k, x_{k+1}, \ldots, x_n).
    \]
    By Theorem~\ref{thm:usi_from_major}, we have
    \[
        \|\diag(x_1-z_1, \ldots, x_k-z_k, x_{k+1}, \ldots, x_n)\| \geq \|\diag(-\gamma, \ldots, -\gamma, x_{k+1}, \ldots, x_n)\| = \|X - Y\|.
    \]
    When combined with Inequality~\eqref{ineq:norm_to_diag}, this gives $\|X - Y\| \le \|X - Z\|$, which completes the proof.
\end{proof}

It's worth noting that if $X \in \D_n$ then constructing the spectral decomposition of $X$ requires full knowledge of $X$, which can be costly when $n$ is large. However, the best rank-$k$ approximation that is described by this theorem can be found, for example, by using power iterations to find the largest eigenvalue $x_1$ and corresponding eigenprojection $\mathbf{v}_1 \mathbf{v}_1^*$, then the second-largest eigenvalue $x_2$ and corresponding eigenprojection $\mathbf{v}_2 \mathbf{v}_2^*$, and so on.

\section{The farthest state from the set of low-rank states}\label{sec:farthest_state}

Now that we know how to find the closest low-rank state to a given state, we explore the question of which state is farthest from the set of low-rank states. We might expect $d(X,\D_{n,k})$ to be maximized when $X = \frac{1}{n}I_n \in \D_n$ (i.e., when $X$ is the maximally-mixed state). It turns out that, for many unitary similarity invariant norms, this is not the case. For example, Theorem~\ref{thm:low_rank} tells us that if $n = 4$, $k = 3$, $\|\cdot\|$ is the operator norm, and $X = \mathrm{diag}\big(\frac{1}{3}, \frac{1}{3}, \frac{1}{3}, 0\big)$, then $d(X,\D_{n,k}) = \frac{1}{3} > \frac{1}{4} = d(\tfrac{1}{n}I_n,\D_{n,k})$.

In this section, we explore which unitary similarity invariant norms lead to $d(X,\D_{n,k})$ being maximized by $X = \frac{1}{n}I$, and where the maximum is attained by the other unitary similarity invariant norms. We start by showing that $X = \frac{1}{n}I_n$ maximizes $d(X,\D_{n,k})$ when the norm in use is the trace norm:

\begin{proposition}\label{prop:max_mixed_trace_norm}
    Suppose $1 \leq k \leq n$ are integers, $\|\cdot\|$ is the Schatten $1$-norm (i.e., the trace norm), and $X \in \D_n$. Then
    \[
        d(X,\D_{n,k}) \leq d(\tfrac{1}{n}I_n,\D_{n,k}) = \frac{2(n-k)}{n}.
    \]
\end{proposition}

\begin{proof}
    The equality on the right follows immediately from Theorem~\ref{thm:low_rank}.

    To prove the inequality, given $X \in \D_n$ with eigenvalues $x_1 \geq x_2 \geq \cdots \geq x_n \geq 0$, let $\gamma$ and $Y$ be as in the statement of Theorem~\ref{thm:low_rank}. That theorem then tells us that
    \begin{align}\label{eq:trace_norm_ineq}
        d(X,\D_{n,k}) = \|X - Y\| = \| \mathrm{diag}(-\gamma,\ldots,-\gamma,x_{k+1},\ldots,x_n) \| = \left(\sum_{j=1}^k \gamma\right) + \left(\sum_{j=k+1}^n x_j\right) = 2\sum_{j=k+1}^n x_j.
    \end{align}
    The facts that $x_1 \geq x_2 \geq \cdots x_n \geq 0$ and $\sum_{j=1}^n x_j = 1$ imply that $\sum_{j=1}^k x_j \geq \frac{k}{n}$, so
    \[
        \sum_{j=k+1}^n x_j = 1 - \sum_{j=1}^k x_j \leq 1 - \frac{k}{n} = \frac{n-k}{n}.
    \]
    Combining this with Equation~\eqref{eq:trace_norm_ineq} gives the desired inequality and completes the proof.
\end{proof}

The situation for general USI norms is somewhat more complicated than it is for the trace norm. Our main result of this section shows that, even though $d(X,\D_{n,k})$ is not necessarily maximized by $X = \tfrac{1}{n}I_n$, it is always maximized by some state of the form $X = \frac{1}{m}I_m \oplus O_{n-m}$, where $m \in \{k+1,k+2,\ldots,n\}$:

\begin{theorem}\label{thm:max_mixed_maximizes}
    Suppose $1 \leq k \leq n$ are integers, $\|\cdot\|$ is a unitary similarity invariant norm, and $X \in \D_n$. Then
    \[
        d(X,\D_{n,k}) \leq \max\big\{ d\big(\tfrac{1}{m} I_m \oplus O_{n-m},\D_{n,k}\big): m \in \{k+1, k+2, \ldots, n\} \big\}.
    \]
\end{theorem}

\begin{proof}
    Denote the eigenvalues of $X$ by $x_1 \geq x_2 \geq \cdots \geq x_n \geq 0$. If $Y \in \D_n$ is the best rank $k$ approximation of $X$ with respect to $\|\cdot\|$ then Theorem~\ref{thm:low_rank} shows that
    \[
        \|X - Y\| = \|\diag(-\gamma, \ldots, -\gamma, x_{k+1}, \ldots, x_n)\|,
    \]
    where $\gamma = \frac{1}{k}\sum_{j=k+1}^n x_j$. Our goal is to determine the value of
    \[
        \max\left\{ \big\|\diag(-\gamma, \ldots, -\gamma, x_{k+1}, \ldots, x_n)\big\| : x_1 \geq x_2 \geq \cdots \geq x_n \geq 0, \sum_{j=1}^n x_j = 1, \gamma = \frac{1}{k}\sum_{j=k+1}^n x_j \right\}.
    \]

    By compactness, the maximum is attained by some $X \in \D_n$, which (by unitary similarity invariance of the norm) we can assume without loss of generality is diagonal. We claim that we may assume that $x_1 = x_2 = \cdots = x_k$. To see why, notice that if this equality does not hold then we can replace $x_i$ by $\frac{1}{k}\sum_{j=1}^k x_j$ for $i \in \{1, 2, \ldots, k\}$ to get a new matrix $\widetilde{X}$. If $\widetilde{Y} \in \D_n$ is the best rank-$k$ approximation of $\widetilde{X}$, then $\widetilde{X} - \widetilde{Y} = X - Y$, so $\|\widetilde{X} - \widetilde{Y}\| = \|X-Y\|$.

    Now, if $x_k > x_{k+1}$ then we may change $(x_k, x_{k+1})$ to $(x_k - d, x_{k+1}+d)$ for $d = \frac{1}{2}(x_k-x_{k+1}) > 0$ to get a matrix $\hat{X}$ (with best rank-$k$ approximation $\hat{Y} \in \D_n$) for which $\|X-Y\| \le \|\hat{X} - \hat{Y}\|$. This follows from the fact that, if we define $\hat{\gamma}$ analogously to how we defined $\gamma$, but coming from the diagonal entries of $\hat{X}$ instead of $X$, then the vector $(-\gamma, \ldots, -\gamma, x_{k+1}, x_{k+2}, \ldots, x_n)$ is majorized by $(-\hat{\gamma}, \ldots, -\hat{\gamma}, x_{k+1} + d, x_{k+2}, \ldots, x_n)$. It follows that the maximum is attained when $X$ has the form $X = \diag(x, \ldots, x, x, x_{k+2}, \ldots, x_n)$ for some real $x$.

    Our next goal is to show that the maximum is attained at a matrix of the form
    \[
        X = \diag(x, \ldots, x, y, 0, \ldots, 0)
    \]
    for some real $x \geq y \geq 0$. To see this, notice that if $X = \diag(x, \ldots, x, x, x_{k+2}, \ldots, x_n)$ is such that there exist integers $k+2 \leq m < q \leq n$ for which $x_{m-1} > x_m \geq x_q > x_{q+1}$, then we can increase $x_m$ and decrease $x_q$ to get a matrix $X^\prime$ (and best rank-$k$ approximation $Y^\prime \in \D_n$) with
    \[
        \|X - Y\| \le \|X^\prime - Y^\prime\|.
    \]
    By repeating this operation, we may bring $X$ to the form
    \[
        X = \diag(\underbrace{x, \ldots, x}_m, y, 0, \ldots, 0)
    \]
    for some integer $m \in \{k+1, k+2, \ldots, n\}$.
    
    Finally, define $X_1 = \frac{1}{m} I_m \oplus O_{n-m}$ and $X_2 = \frac{1}{m+1} I_{m+1} \oplus O_{n-(m+1)}$. Then $Y = \frac{1}{k} I_k \oplus O_{n-k} \in \D_n$ is the best rank-$k$ approximation for each of $X$, $X_1$, and $X_2$. Since $X - Y$ is a convex combination of $X_1 - Y$ and $X_2 - Y$, we have
    \[
        d(X,\D_{n,k}) = \|X-Y\| \leq \max\big\{\|X_1-Y\|, \|X_2-Y\|\big\} = \max\big\{d(X_1,\D_{n,k}), d(X_2,\D_{n,k})\big\},
    \]
    which completes the proof.
\end{proof}

The remainder of this section is devoted to exploring the various consequences of Theorem~\ref{thm:max_mixed_maximizes}, and in particular exploring in which special cases $d(X,\D_{n,k})$ is maximized by the maximally-mixed state (other than the case when the norm is the trace norm, as established by Proposition~\ref{prop:max_mixed_trace_norm}). We start with the simple observation that $k \geq n - 1$ is one such case:

\begin{corollary}\label{cor:max_mixed_nminone}
    Suppose $n$ is an integer, $\|\cdot\|$ is a unitary similarity invariant norm, and $X \in \D_n$. If $k \geq n-1$ then
    \[
        d(X,\D_{n,k}) \leq d(\tfrac{1}{n}I_n,\D_{n,k}).
    \]
\end{corollary}

\begin{proof}
    The $k = n$ case is trivial, since $d(X,\D_{n,k}) = d(\tfrac{1}{n}I_n,\D_{n,k}) = 0$. The $k = n-1$ case follows immediately from Theorem~\ref{thm:max_mixed_maximizes}, since the only choice of $m \in \{k+1, k+2,\ldots,n\}$ is $m = n$.
\end{proof}

The maximally-mixed state is also always the farthest from the set of low-rank states in small dimensions:

\begin{proposition}\label{prop:max_mixed_n3}
    Suppose $1 \leq k \leq 3$ is an integer, $\|\cdot\|$ is a unitary similarity invariant norm, and $X \in \D_3$. Then
    \[
        d(X,\D_{3,k}) \leq d(\tfrac{1}{3}I_3,\D_{3,k}).
    \]
\end{proposition}

\begin{proof}
    Denote the eigenvalues of $X$ by $x_1 \geq x_2 \geq x_3 \geq 0$ and let $\gamma$ and $Y$ be as in the statement of Theorem~\ref{thm:low_rank}. If $k \geq 2$ then the result follows immediately from Corollary~\ref{cor:max_mixed_nminone}, so we assume from now on that $k = 1$.
    
    Note that Theorem~\ref{thm:low_rank} tells us that
    \[
        d(\tfrac{1}{3}I_3,\D_{3,1}) = \left\| \mathrm{diag}\left(\frac{-2}{3}, \frac{1}{3}, \frac{1}{3}\right) \right\| \quad \text{and} \quad d(X,\D_{3,1}) = \left\| \mathrm{diag}\left(-x_2-x_3, x_2, x_3\right) \right\|.
    \]
    For $0 \leq t \leq 1$, we have
    \begin{align}\label{eq:t_major_vec}
        t\left(\frac{1}{3}, \frac{1}{3}, \frac{-2}{3}\right) + (1-t)\left(\frac{2}{3}, \frac{-1}{3}, \frac{-1}{3}\right) = \frac{1}{3}(2-t, 2t-1, -t-1).
    \end{align}
    If we can find a value of $t$ for which the vector $(x_2, x_3, -x_2-x_3)$ is majorized by the vector~\eqref{eq:t_major_vec}, then Theorem~\ref{thm:usi_from_major} will imply $d(X,\D_{3,k}) \leq d(\tfrac{1}{3}I_3,\D_{3,k})$. If $x_2 \leq 1/3$ then we can choose $t = 1$ and if $1/3 \leq x_2 \leq 1/2$ then we can choose $t = 2 - 3x_2$. Since $x_1 + x_2 \leq 1$ we must have $x_2 \leq 1/2$, so this completes the proof.
\end{proof}

\subsection{The Schatten norms}\label{sec:farthest_schatten}

Next, we characterize which Schatten $p$-norms have the property that $d(X,\D_{n,k})$ is always maximized by the maximally-mixed state:

\begin{theorem}\label{thm:schatten_124}
    Suppose $p \in [1,\infty]$ and $\|\cdot\|$ is the Schatten $p$-norm. The following are equivalent:
    \begin{itemize}
        \item $d(X,\D_{n,k}) \leq d(\tfrac{1}{n}I_n,\D_{n,k})$ for all integers $1 \leq k \leq n$ and all $X \in \D_n$.

        \item $p \in \{1\} \cup [2,4]$.
    \end{itemize}
\end{theorem}

Before proving this theorem, we note that Theorem~\ref{thm:low_rank} tells us that the upper bound in the first bullet point is equal to
\[
    d(\tfrac{1}{n}I_n,\D_{n,k}) = \frac{n-k}{n^p} + k\left(\frac{1}{k} - \frac{1}{n}\right)^p.
\]

\begin{proof}[Proof of Theorem~\ref{thm:schatten_124}]
    The $p = 1$ case is established by Proposition~\ref{prop:max_mixed_trace_norm}.

    For $1 < p < 2$, we need to find $1 \leq k \leq n$ and $X \in \D_n$ for which $d(X,\D_{n,k}) > d(\tfrac{1}{n}I_n,\D_{n,k})$. To this end, let $k = 1$ and $X = \mathrm{diag}(\tfrac{1}{n-1}I_{n-1}, 0)$. By Theorem~\ref{thm:low_rank}, we have
    \begin{align*}
        d(\tfrac{1}{n}I_n,\D_{n,k}) & = \frac{(n-1) + (n-1)^p}{n^p} \quad \text{and} \quad d(X,\D_{n,k}) = \frac{(n-2) + (n-2)^p}{(n-1)^p}.
    \end{align*}
    In particular, if we define a function $f : [0,\infty) \rightarrow [0,\infty)$ by $f(x) = (x + x^p)/(x+1)^p$ then $d(\tfrac{1}{n}I_n,\D_{n,k}) = f(n-1)$ and $d(X,\D_{n,k}) = f(n-2)$. We compute
    \[
        x(x+1)^{p+1}f^\prime(x) = px^p + x - (p-1)x^2.
    \]
    For $1 < p < 2$, the term $(p-1)x^2$ grows quicker than $px^p + x$, so $f^\prime(x) < 0$ for sufficiently large $x \geq 0$. This implies
    \[
        d(\tfrac{1}{n}I_n,\D_{n,k}) = f(n-1) <  f(n-2) = d(X,\D_{n,k})
    \]
    for all sufficiently large $n$, which finishes this case.

    For $2 \leq p \leq 4$, Theorem~\ref{thm:max_mixed_maximizes} tells us that $d(X,\D_{n,k})$ is maximized by a density matrix of the form $X = \mathrm{diag}(\frac{1}{m}I_m, O_{n-m})$ for some $m \in \{k+1,k+2,\ldots,n\}$. Theorem~\ref{thm:low_rank} tells us that, for this $X$, we have
    \[
        d(X,\D_{n,k}) = \frac{m-k}{m^p} + k\left(\frac{1}{k} - \frac{1}{m}\right)^p.
    \]
    Our goal is to show that, for $2 \leq p \leq 4$, this quantity is maximized when $m = n$.

    To this end, for fixed $1 \leq k < n$ define $f : [1,k] \rightarrow [0,\infty)$ by
    \[
        f(x,p) = d(I_n/n,\D_{n,k}) - d(X,\D_{n,k}) = \left(\frac{n-k}{n^p} + k\left(\frac{1}{k} - \frac{1}{n}\right)^p\right) - \left(\frac{x-k}{x^p} + k\left(\frac{1}{k} - \frac{1}{x}\right)^p\right).
    \]
    Our goal is to show that $f(x,p) \geq 0$ for all $2 \leq p \leq 4$ and $k+1 \leq x \leq n$. For $p = 2$, we have $f(x,2) = 1/x - 1/n \geq 0$ for $x \leq n$. For $2 \leq p < 4$, we compute $f(n,p) = 0$ and
    \begin{align}\label{eq:p24_deriv}
        k^{p-2} x^{p+1}\left(\frac{d}{dx}f(x,p)\right) = -p(x-k)^{p-1} - k^{p-2}(pk - (p-1)x).
    \end{align}
    If we can show that the quantity on the right of Equation~\eqref{eq:p24_deriv} (which we will denote by $g(x)$ and think of solely as a function of $x$) is non-positive then we will be done. To this end, notice that $g(k) = -k^{p-1} < 0$, $\lim_{x\rightarrow\infty} g(x) = -\infty$, and $g$ has a unique critical point in $[k,\infty)$, located at
    \[
        \tilde{x} = k\left(1 + \frac{1}{p^{1/(p-2)}}\right).
    \]
    If we can show that $g(\tilde{x}) \leq 0$ then that would imply $g(x) \leq 0$ for all $x \in [k,\infty)$ and thus for all $x \in [k+1,n]$, as desired.

    To show that $g(\tilde{x}) \leq 0$, we compute
    \[
        k^{p-1}g(\tilde{x}) = \frac{p-2}{p^{1/(p-2)}} - 1.
    \]
    Standard calculus techniques show that $p-2 \leq p^{1/(p-2)}$ when $p \in (2,4]$, so $g(\tilde{x}) \leq 0$, completing the $2 \leq p \leq 4$ case.

    Finally, for $p > 4$, our goal is to find $1 \leq k \leq n$ and $X \in \D_n$ for which $d(X,\D_{n,k}) > d(I/n,\D_{n,k})$. Let $m \geq 2$ be an integer, $n = 3m$, $k = 2m$, and let $\|\cdot\|$ be the Schatten $p$-norm (for some $p \in [1,\infty)$). Applying Theorem~\ref{thm:low_rank} shows that
    \[
        d(\tfrac{1}{n}I_n,\D_{n,k}) = \frac{m(2^p + 2)}{(6m)^p}.
    \]
    If we define $X = \mathrm{diag}\big(\frac{1}{3m-1}I_{3m-1}, 0\big) \in \D_n$ then applying Theorem~\ref{thm:low_rank} again shows that
    \[
        d(X,\D_{n,k}) = \frac{2m(m-1)^p + (2m)^p(m-1)}{m^p(6m-2)^p}.
    \]
    Then
    \begin{align}\begin{split}\label{eq:mu_lim}
        & \quad \ \ 2^p(3m-1)^p(d(X,\D_{n,k}) - d(\tfrac{1}{n}I_n,\D_{n,k})) \\
        & = (m-1)2^p + 2m\left(1-\frac{1}{m}\right)^p - m(2 + 2^p)\left(1-\frac{1}{3m}\right)^p \\
        & \geq (m-1)2^p + 2m\left(1-\frac{p}{m}\right) - m(2 + 2^p)\left(1-\frac{p}{3m} + \frac{p(p-1)}{18m^2}\right),
    \end{split}\end{align}
    where the inequality holds when $m$ is sufficiently large; a fact that comes from applying the binomial series to $\left(1 - \frac{1}{m}\right)^p$ and $\left(1 - \frac{1}{3m}\right)^p$. Taking the limit as $m \rightarrow \infty$ in the quantity on the right of Equation~\eqref{eq:mu_lim} gives
    \[
        \frac{(p - 3)2^p - 4p}{3},
    \]
    which is strictly positive whenever $p > 4$. In particular, this means that for every $p > 4$, there exists a sufficiently large $m$ for which $d(X,\D_{n,k}) > d(\tfrac{1}{n}I_n,\D_{n,k})$.
\end{proof}

It is worth emphasizing that Theorem~\ref{thm:schatten_124} only establishes which Schatten $p$-norms are such that $d(X,\D_{n,k})$ is maximized by the maximally-mixed state for all $n$ and $k$; it does not say which Schatten $p$-norms have this property for fixed $k$ and $n$. For example, Proposition~\ref{prop:max_mixed_n3} showed that when $n = 3$ every unitary similarity invariant norm has this property (not just the Schatten $p$-norms, and certainly not just the Schatten $p$-norms with $p \in \{1\} \cup [2,4]$).

When $n$ and $k$ are given, determining exactly which Schatten $p$-norms have $d(X,\D_{n,k})$ maximized by the maximally-mixed state seems to be tricky without explicitly computing all $n-k$ values of $d(I_m/m \oplus O_{n-m},\D_{n,k})$ for $m \in \{k+1, k+2, \ldots, n\}$. For example, if $n = 14$ and $k = 9$, then direct computation shows that the optimal choice of $m$ in Theorem~\ref{thm:max_mixed_maximizes} for the Schatten $p$-norms are as follows:
\begin{itemize}
    \item $m = 14$ if $p \in [1, \alpha_1)$, where $\alpha_1 \approx 4.00865$ is the unique positive real number satisfying \[ 13^{\alpha_1}(9\cdot 5^{\alpha_1} + 5\cdot 9^{\alpha_1}) = 14^{\alpha_1} (9\cdot 4^{\alpha_1} + 4\cdot 9^{\alpha_1}).\]
    
    \item $m = 13$ if $p \in [\alpha_1, \alpha_2)$, where  $\alpha_2 \approx 4.14468$ is the unique positive real number satisfying \[ 3\cdot 13^{\alpha_2} (3^{\alpha_2} + 3) = 4^{\alpha_2} (9\cdot 4^{\alpha_2} + 4\cdot 9^{\alpha_2}).\]
    
    \item $m = 12$ if $p \in [\alpha_2, \alpha_3)$, where $\alpha_3 \approx 4.79781$ is the unique positive real number satisfying \[ 3\cdot 11^{\alpha_3} (3^{\alpha_3} + 3) = 4^{\alpha_3} (9\cdot 2^{\alpha_3} + 2\cdot 9^{\alpha_3}).\]
    
    \item $m = 11$ if $p \in [\alpha_3, \alpha_4)$, where $\alpha_4 \approx 7.27337$ is the unique positive real number satisfying \[ 11^{\alpha_4} (9^{\alpha_4} + 9) = 10^{\alpha_4} (9\cdot 2^{\alpha_4} + 2\cdot 9^{\alpha_4}),\]and
    
    \item $m = 10$ if $p \in [\alpha_4, \infty]$.
\end{itemize}

\subsection{The Ky Fan norms}\label{sec:farthest_ky_fan}

We now explore the question of which state $X \in \D_n$ is farthest from the set of rank-at-most-$k$ states when we measure distance via one of the Ky Fan norms. Theorem~\ref{thm:max_mixed_maximizes} of course applies in this setting and tells us that the farthest state is one of $n-k$ different candidates, corresponding to the $n - k$ different choices of $m \in \{k+1, k+2, \ldots, n\}$. It follows that we could always just check which of those $n-k$ different candidates is farthest from $\D_{n,k}$.

Our main result of this section shows that it actually always suffices to check $5$ candidates, since the optimal candidate has $m \in \left\{r, r+k, n, \left\lfloor \sqrt{k(2k+r)} \right\rfloor, \left\lceil \sqrt{k(2k+r)} \right\rceil\right\}$. Furthermore, for many choices of $k$ and Ky Fan norm, we can even state the optimal candidate explicitly (i.e., there is just one choice of $m$):

\begin{theorem}\label{thm:ky_fan_extreme}
    Suppose $1 \leq k,r \leq n$ are integers, $\|\cdot\|$ is the Ky Fan $r$-norm, and $X \in \D_n$. Define $g(r,n) = \frac{1}{2}\big(\sqrt{r(4n+r)}-r\big)$. Then $d(X,\D_{n,k}) \leq d(\tfrac{1}{m}I_m \oplus O_{n-m},\D_{n,k})$, where
    \begin{align*}
        m = \begin{cases}
            r, & \text{if \ \ $k \leq \frac{r}{2}$}, \\
            n, & \text{if \ \ $\frac{r}{2} < k \leq r$}, \\
            n, & \text{if \ \ $\left(\frac{1+\sqrt{5}}{2}\right)r < k$ and $k < g(n,r)$}, \\
            \min\{r+k,n\}, & \text{if \ \ $\left(\frac{1+\sqrt{5}}{2}\right)r < k$ and $g(r,n) \leq k$}.
        \end{cases}
    \end{align*}
    Furthermore, if $r < k \leq \left(\frac{1+\sqrt{5}}{2}\right)r$ (i.e., in all cases not covered above) then
    \[
        m \in \left\{ n, k+r, \left\lfloor \sqrt{k(2k+r)} \right\rfloor, \left\lceil \sqrt{k(2k+r)} \right\rceil \right\}.
    \]
\end{theorem}

\begin{proof}
    By Theorem~\ref{thm:max_mixed_maximizes}, we know that $d(X,\D_{n,k}) \leq d(\tfrac{1}{m}I_m \oplus O_{n-m},\D_{n,k})$ for some $m \in \{k+1, k+2, \ldots, n\}$; we just need to show that $m = \min\{k+r, n\}$ gives the largest right-hand side.

    To this end, note that Theorem~\ref{thm:low_rank} tells us that, for each $m \geq k+1$, we have
    \begin{align*}
        d(\tfrac{1}{m}I_m \oplus O_{n-m},\D_{n,k}) & = \Bigg\|\diag\Bigg(\underbrace{\frac{1}{m}-\frac{1}{k}, \ldots, \frac{1}{m}-\frac{1}{k}}_k, \underbrace{\frac{1}{m}, \ldots, \frac{1}{m}}_{m-k}, \underbrace{0, \ldots, 0}_{n-m}\Bigg)\Bigg\|_{(r)}.
    \end{align*}
    We now split into three cases, depending on the value of $k$.

    Case 1: $k \leq r/2$. Then
    \[
        d(\tfrac{1}{m}I_m \oplus O_{n-m},\D_{n,k}) = \begin{cases}
            2 - \frac{2k}{m}, & \text{if \ \ $m \leq r$}, \\
            1 + \frac{r-2k}{m}, & \text{if \ \ $m \geq r$}.
        \end{cases}
    \]
    If $m \leq r$ then $2 - \frac{2k}{m}$ is maximized by choosing $m$ as large as possible (i.e., $m = r$) and if $m \geq r$ then (since $k \leq r/2$), $1 + \frac{r-2k}{m}$ is maximized by choosing $m$ as small as possible (i.e., $m = r$). Either way, the optimal choice is $m = r$, which finishes this case.

    Case 2: $r/2 < k \leq r$. Then
    \[
        d(\tfrac{1}{m}I_m \oplus O_{n-m},\D_{n,k}) = \begin{cases}
            2 - \frac{2k}{m}, & \text{if \ \ $m \leq r$}, \\
            1 - \frac{2k-r}{m}, & \text{if \ \ $m \geq r$}.
        \end{cases}
    \]
    Both branches of this function are maximized by choosing $m$ as large as possible, so we choose $m = r$ in the first branch and $m = n$ in the second branch. The inequality $1 - \frac{2k-r}{n} \geq 2 - \frac{2k}{r}$ is equivalent to the clearly true inequality $(2k-r)(n-r) \geq 0$, so the maximum occurs when $m = n$ (in the second branch).

    Case 3: $r < k$. Then
    \begin{align}\begin{split}\label{eq:case3_rk}
        d(\tfrac{1}{m}I_m \oplus O_{n-m},\D_{n,k}) = \begin{cases}
            r\left(\frac{1}{k} - \frac{1}{m}\right), & \text{if \ \ $2k \leq m$}, \\
            \frac{r}{m}, & \text{if \ \ $r+k \leq m \leq 2k$}, \\
            \frac{(m-k)(2k+r-m)}{km}, & \text{if \ \ $m \leq r+k$}. \\
        \end{cases}
    \end{split}\end{align}
    The first branch of this function is maximized when $m$ is as large as possible, so $m = n$. The second branch is maximized when $m$ is as small as possible, so $m = r+k$. The third branch is more complicated. Standard calculus techniques show that it has a unique critical point, which is the global maximizer, at $m = \sqrt{k(2k+r)}$ (this value of $m$ is not necessarily an integer; this will be dealt with in the remainder of the proof). We now split into two  sub-cases:

    \begin{itemize}
        \item Case 3a: $\sqrt{k(2k+r)} > r+k$ (i.e., $\left(\frac{1+\sqrt{5}}{2}\right)r < k$). Then the third branch of~\eqref{eq:case3_rk} is an increasing function of $m$ for all $m \leq r+k$ (where the branch is valid), so the optimal choice is $m = r+k$. It is straightforward to show that if $k \geq g(r,n) = \frac{1}{2}\big(\sqrt{r(4n+r)}-r\big)$ then $\frac{r}{r+k}$ (from the second or third branch) is larger than $r\left(\frac{1}{k} - \frac{1}{n}\right)$ (from the first branch) so $m = r+k$ is the optimal choice, and if $k < g(r,n)$ then the opposite is true so $m = n$ is the optimal choice.

        \item Case 3b: $\sqrt{k(2k+r)} \leq r+k$ (i.e., $k \leq \left(\frac{1+\sqrt{5}}{2}\right)r$). Then the third branch of~\eqref{eq:case3_rk} is an increasing function of $m$ for $m \leq \sqrt{k(2k+r)}$ and a decreasing function of $m$ for $m \geq \sqrt{k(2k+r)}$. The optimal integer value of $m$ is thus either $m = \lfloor \sqrt{k(2k+r)}\rfloor$ or $m = \lceil \sqrt{k(2k+r)}\rceil$ (note that since $\sqrt{k(2k+r)} \leq r+k$, we have $\lceil \sqrt{k(2k+r)}\rceil \leq r+k$ too, so $m = \lceil \sqrt{k(2k+r)}\rceil$ is a valid choice). It follows that one of these two values is the maximizer in the third branch of~\eqref{eq:case3_rk}. Since $m = n$ is the maximizer for the first branch and $m = r+k$ is the maximizer for the second branch, the proof is complete.
    \end{itemize}
\end{proof}

We note that it is possible for the optimal $m$ to be one of $\lfloor \sqrt{k(2k+r)}\rfloor$ or $\lceil \sqrt{k(2k+r)}\rceil$, so it seems unlikely that this result can be improved much. The earliest example of this type occurs when $n = 9$, $r = 4$, and $k = 5$: the optimal value of $m$ in this case is $m = \lfloor \sqrt{k(2k+r)}\rfloor = 8$.

When $r = n$, Theorem~\ref{thm:ky_fan_extreme} tells us that the optimal value of $m$ is $m = n$ (regardless of $k$), thus recovering Proposition~\ref{prop:max_mixed_trace_norm}. When $r = 1$, it tells us the maximum possible value of $d(X,\D_{n,k})$ under the operator norm:

\begin{corollary}\label{cor:max_mixed_op_norm}
    Suppose $1 \leq k < n$ are integers, $\|\cdot\|$ is the operator norm, and $X \in \D_n$. Then
    \[
        d(X,\D_{n,k}) \leq \begin{cases}
            d(\tfrac{1}{n}I_n,\D_{n,k}) = \frac{1}{k} - \frac{1}{n} & \text{if \ \ $k(k+1) \leq n$,} \\
            d(\tfrac{1}{k+1}I_{k+1} \oplus O_{n-(k+1)},\D_{n,k}) = \frac{1}{k+1} & \text{otherwise}.
        \end{cases}
    \]
\end{corollary}

\begin{proof}
    If $k = 1$ then the result follows from the ``$r/2 < k \leq r$'' branch of Theorem~\ref{thm:ky_fan_extreme}. If $k \geq 2$ then (since $(1+\sqrt{5})/2 < 2$) we have $((1+\sqrt{5})/2)r < k$. Similarly, if $k(k+1) < n$ then (in the notation of Theorem~\ref{thm:ky_fan_extreme}) we have $k < g(n,r)$. The $k \geq 2$, $k(k+1) \leq n$ portion of this corollary thus follows from the ``$((1+\sqrt{5})/2)r < k$ and $k < g(n,r)$ branch of Theorem~\ref{thm:ky_fan_extreme}.

    The ``otherwise'' branch of this corollary follows from the ``$((1+\sqrt{5})/2)r < k$ and $g(n,r) \leq k$ branch of Theorem~\ref{thm:ky_fan_extreme}.
\end{proof}

\section*{Acknowledgements}

The authors thank Matthew Lin for helpful conversations related to this work. N.J.\ was supported by NSERC Discovery Grant number RGPIN-2022-04098. C.-K.~L. is an affiliate member of the Institute for Quantum Computing of the University of Waterloo, and also an affiliate member of the Quantum Science \& Engineering Center of the George Mason University; his research was partially supported by the Simons Foundation Grant 851334.

\bibliographystyle{alpha}
\bibliography{ref}

\end{document}